\documentclass[a4paper,11pt]{article}

\usepackage[left=2.5cm,right=2.5cm,top=3cm,bottom=3cm,pdftex]{geometry}
\usepackage{amssymb}
\usepackage{amsthm}
\usepackage{amsmath}

\usepackage{url}
\usepackage{natbib}

\usepackage[utf8]{inputenc}
\usepackage[T1]{fontenc}

\usepackage[pdftex]{graphicx}
\DeclareGraphicsExtensions{.png, .pdf, .jpg} 

\usepackage[pdftex, colorlinks, linkcolor=blue, urlcolor=blue, citecolor=blue, pagecolor=blue, breaklinks=true]{hyperref}

\newtheorem{theorem}{Theorem}

\usepackage{mathtools}
\usepackage{amsfonts}
\usepackage{mathrsfs}
\usepackage{bbm}
\usepackage{bm}
\usepackage{mathabx}

\usepackage{algorithm}
\usepackage{algorithmic}
\usepackage{booktabs}
\usepackage{multirow}
\usepackage{enumerate}
\usepackage{setspace,supertabular}
\usepackage{longtable}
\usepackage{multirow}
\usepackage{lscape}
\usepackage{hyperref}

\usepackage[table]{xcolor}
\usepackage{array}

\newcolumntype{x}[1]{%
>{\raggedleft\hspace{0pt}}p{#1}}%

\floatname{algorithm}{Procedure}

\newtheorem{assump}{Assumption}

\newcommand{\eps}{\varepsilon}

\newcommand{\bLambda}{\bm{\Lambda}}

\newcommand{\E}{\mathsf{E}}

\newcommand{\Var}{\mathsf{Var}\,}
\newcommand{\Cov}{\mathsf{Cov}\,}
\newcommand{\Corr}{\mathsf{Corr}\,}
\newcommand{\prob}{\mathsf{P}}

\newcommand{\dist}{\mathscr{D}}

\newcommand{\ud}{\mbox{d}}

\definecolor{seda}{gray}{0.9}
\definecolor{sedaseda}{gray}{0.8}

\begin{document}
\pagestyle{empty} 

\begin{center}
\Large\textbf{Change Point in Panel Data with Small Fixed Panel Size:\\Ratio and Non-Ratio Test Statistics} \\[11pt]
\normalsize
Barbora Pe\v{s}tov\'{a}$^{1}$, Michal Pe\v{s}ta$^{2}$\footnote{Corresponding author: \href{mailto:michal.pesta@mff.cuni.cz}{michal.pesta@mff.cuni.cz}}\\[11pt]

\footnotesize
1. Institute of Computer Science, The Czech Academy of Sciences, Department of Medical Informatics and Biostatistics, Pod Vod\'{a}renskou v\v{e}\v{z}\'{i}~271/2, 18207 Prague, Czech Republic\\
2. Faculty of Mathematics and Physics, Charles University in Prague, Department of Probability and Mathematical Statistics, Sokolovsk\'{a}~49/83, 18675 Prague, Czech Republic\\
\normalsize
\end{center}

\begin{quote}
\textbf{Abstract:} The main goal is to develop and, consequently, compare stochastic methods for detection whether a~structural change in panel data occurred at some unknown time or not. Panel data of our interest consist of a~moderate or relatively large number of panels, while the panels contain a~small number of observations. Testing procedures to detect a~possible common change in means of the panels are established. Ratio and non-ratio type test statistics are considered. Their asymptotic distributions under the no change null hypothesis are derived. Moreover, we prove the consistency of the tests under the alternative. The main advantage of the ratio type statistics compared to the non-ratio ones is that the variance of the observations neither has to be known nor estimated. A~simulation study reveals that the proposed ratio statistic outperforms the non-ratio one by keeping the significance level under the null, mainly when stronger dependence within the panel is taken into account. However, the non-ratio statistic rejects the null in the simulations more often than it should, which yields higher power compared to the ratio statistic.

\end{quote}

\section{Introduction}\label{sec:intro}
The problem of an unknown common change in means of the panels is studied here, where the panel data consist of $N$ panels and each panel contains $T$ observations over time. Various values of the change are possible for each panel at some unknown common time $\tau=1,\ldots,N$. The panels are considered to be independent, but this restriction can be weakened. In spite of that, observations within the panel are usually not independent. It is supposed that a~common unknown dependence structure is present over the panels.

Tests for change point detection in the panel data have been proposed only in case when the panel size $T$ is sufficiently large, i.e., $T$ increases over all limits from an asymptotic point of view, cf.~\cite{CHH2013} or~\cite{HH2012}. However, the change point estimation has already been studied for finite $T$ not depending on the number of panels $N$, see~\cite{Bai2010} or~\cite{PP2016}. The remaining task is to develop testing procedures to decide whether a~common change point is present or not in the panels, while taking into account that the length $T$ of each observation regime is fixed and can be relatively small.




\section{Panel Change Point Model}\label{sec:model}
Let us consider the panel change point model
\begin{equation}\label{model}
Y_{i,t}=\mu_i+\delta_i\mathcal{I}\{t>\tau\}+\sigma\eps_{i,t},\quad 1\leq i\leq N,\, 1\leq t\leq T;
\end{equation}
where $\sigma>0$ is an unknown variance-scaling parameter and $T$ is fixed, not depending on $N$. The possible \emph{common change point time} is denoted by $\tau\in\{1,\ldots,T\}$. A~situation where $\tau=T$ corresponds to \emph{no change} in means of the panels. The means $\mu_i$ are panel-individual. The amount of the break in mean, which can also differ for every panel, is denoted by $\delta_i$. Furthermore, it is assumed that the sequences of panel disturbances $\{\eps_{i,t}\}_t$ are independent and within each panel the errors form a~weakly stationary sequence with a~common correlation structure. This can be formalized in the following assumption.

\begin{assump}\label{ass:A1}
\normalfont The vectors $[\eps_{i,1},\ldots,\eps_{i,T}]^{\top}$ existing on a~probability space $(\Omega,\mathcal{F},\prob)$ are $iid$ for $i=1,\ldots,N$ with $\E\eps_{i,t}=0$ and $\Var\eps_{i,t}=1$, having the autocorrelation function
\[
\rho_t=\Corr\left(\eps_{i,s},\eps_{i,s+t}\right)=\Cov\left(\eps_{i,s},\eps_{i,s+t}\right),\quad\forall s\in\{1,\ldots,T-t\},
\]
which is independent of the lag $s$, the cumulative autocorrelation function
\[
r(t)=\Var\sum_{s=1}^t \eps_{i,s}=\sum_{|s|<t}(t-|s|)\rho_s,
\]
and the shifted cumulative correlation function
\[
R(t,v)=\Cov\left(\sum_{s=1}^t\eps_{i,s},\sum_{u=t+1}^v\eps_{i,u}\right)=\sum_{s=1}^t\sum_{u=t+1}^v\rho_{u-s},\quad t<v
\]
for all $i=1,\ldots,N$ and $t,v=1,\ldots,T$.
\end{assump}

The sequence $\{\eps_{i,t}\}_{t=1}^T$ can be viewed as a~part of a~\emph{weakly stationary} process. Note that the dependent errors within each panel do not necessarily need to be linear processes. For example, GARCH processes as error sequences are allowed as well. The assumption of independent panels can indeed be relaxed, but it would make the setup much more complex. Consequently, probabilistic tools for dependent data need to be used (e.g., suitable versions of the central limit theorem). Nevertheless, assuming, that the claim amounts for different insurance companies are independent, is reasonable. Moreover, the assumption of a~common homoscedastic variance parameter $\sigma$ can be generalized by introducing weights $w_{i,t}$, which are supposed to be known. Being particular in actuarial practice, it would mean to normalize the total claim amount by the premium received, since bigger insurance companies are expected to have higher variability in total claim amounts paid.

It is required to test the \emph{null hypothesis} of no change in the means
\[
H_0:\,\tau=T
\]
against the~\emph{alternative} that at least one panel has a~change in mean
\[
H_1:\,\tau<T\quad\mbox{and}\quad\exists i\in\{1,\ldots,N\}:\,\delta_i\neq 0.
\]

\section{Test Statistic and Asymptotic Results}\label{sec:results}
We propose a~\emph{ratio type statistic} to test $H_0$ against $H_1$, because this type of statistic does not require estimation of the nuisance parameter for the variance. Generally, this is due to the fact that the variance parameter simply cancels out from the nominator and denominator of the statistic. 
For surveys on ratio type test statistics, we refer to \cite{CH1997}, \cite{Barborka2011}, and~\cite{HHH2008}. Our particular panel change point CUSUM test statistic is
\[
\mathcal{C}_N(T)=\frac{1}{\sqrt{N}}\max_{t=1,\ldots,T-1}\left|\sum_{i=1}^N\sum_{s=1}^t\left(Y_{i,s}-\widebar{Y}_{i,T}\right)\right|,
\]
which is going to be compared with the ratio test statistics
\[
\mathcal{R}_N(T)=\max_{t=2,\ldots,T-2}\frac{\max_{s=1,\ldots,t}\left|\sum_{i=1}^N\sum_{r=1}^s\left(Y_{i,r}-\widebar{Y}_{i,t}\right)\right|}{\max_{s=t,\ldots,T-1}\left|\sum_{i=1}^N\sum_{r=s+1}^T\left(Y_{i,r}-\widetilde{Y}_{i,t}\right)\right|},
\]
where $\widebar{Y}_{i,t}$ is the average of the first $t$ observations in panel $i$ and $\widetilde{Y}_{i,t}$ is the average of the last $T-t$ observations in panel $i$, i.e.,
\[
\widebar{Y}_{i,t}=\frac{1}{t}\sum_{s=1}^t Y_{i,s}\quad\mbox{and}\quad\widetilde{Y}_{i,t}=\frac{1}{T-t}\sum_{s=t+1}^T Y_{i,s}.
\]

elaborated in~\cite{PP2015}. It will be demonstrated by simulations that $\mathcal{R}_N(T)$ keeps the theoretical significance level, however, $\mathcal{C}_N(T)$ does not.


Firstly, we derive the behavior of the test statistics under the null hypothesis.

\begin{theorem}[Under Null]\label{underNull}
Under hypothesis $H_0$ and Assumption~\ref{ass:A1}
\[
\mathcal{C}_N(T)\xrightarrow[N\to\infty]{\dist}\sigma\max_{t=1,\ldots,T-1}\left|X_t-\frac{t}{T}X_T\right|
\]
and
\[
\mathcal{R}_N(T)\xrightarrow[N\to\infty]{\dist}\max_{t=2,\ldots,T-2}\frac{\max_{s=1,\ldots,t}\left|X_s-\frac{s}{t}X_t\right|}{\max_{s=t,\ldots,T-1}\left|Z_s-\frac{T-s}{T-t}Z_t\right|},
\]
where $Z_t:=X_T-X_t$ and $[X_1,\ldots,X_T]^{\top}$ is a~multivariate normal random vector with zero mean and covariance matrix $\bLambda=\{\lambda_{t,v}\}_{t,v=1}^{T,T}$ such that
\[
\lambda_{t,t}=r(t)\quad\mbox{and}\quad\lambda_{t,v}=r(t)+R(t,v),\,\, t<v.
\]
\end{theorem}

The limiting distribution does not depend on the variance nuisance parameter $\sigma$, but it depends on the unknown correlation structure of the panel change point model, which has to be estimated for testing purposes. The way of its estimation is shown in Section~\ref{sec:covest}. Note that in case of independent observations within the panel, the correlation structure and, hence, the covariance matrix $\bLambda$ is simplified such that $r(t)=t$ and $R(t,v)=0$.

Next, we show how the test statistic behaves under the alternative.

\begin{assump}\label{alternativeDelta}
\normalfont $\lim_{N\to\infty}\frac{1}{\sqrt{N}}\left|\sum_{i=1}^N\delta_i\right|=\infty$.
\end{assump}

\begin{theorem}[Under Alternative]\label{underAlternative}
If $\tau\leq T-3$, then under Assumptions~\ref{ass:A1}, \ref{alternativeDelta} and alternative $H_1$
\begin{equation}
\mathcal{C}_{N}(T)\xrightarrow[N\to\infty]{\prob}\infty\quad\mbox{and}\quad\mathcal{R}_{N}(T)\xrightarrow[N\to\infty]{\prob}\infty.
\end{equation}
\end{theorem}

Assumption~\ref{alternativeDelta} is satisfied, for instance, if $0<\delta\leq\delta_i\,\forall i$ (a~common lower change point threshold) and $\delta\sqrt{N}\to\infty,\, N\to\infty$. Another suitable example of $\delta_i$s for the condition in Assumption~\ref{alternativeDelta}, can be $0<\delta_i=KN^{-1/2+\eta}$ for some $K>0$ and $\eta>0$. Or $\delta_i=Ci^{\alpha-1}\sqrt{N}$ may be used as well, where $\alpha\geq 0$ and $C>0$. The assumption $\tau\leq T-3$ means that there are at least three observations in the panel after the change point. It is also possible to redefine the test statistic by interchanging the nominator and the denominator of $\mathcal{S}_{N}(T)$. Afterwards, Theorem~\ref{underAlternative} for the modified test statistic would require three observations before the change point, i.e., $\tau\geq 3$.

Theorem~\ref{underAlternative} says that in presence of a~structural change in the panel means, the test statistic explodes above all bounds. Hence, the procedure is consistent and the asymptotic distribution from Theorem~\ref{underNull} can be used to construct the test.

\section{Estimation of the Covariance Structure}\label{sec:covest}
The estimation of the covariance matrix $\bLambda$ from Theorem~\ref{underNull} requires panels as vectors with elements having common mean (i.e., without a~jump). Therefore, it is necessary to construct an~estimate for a~possible change point. A~\emph{consistent estimate} of the change point $\tau$ in the panel data is proposed in~\cite{PP2016} as
\begin{equation}\label{tauhat}
\widehat{\tau}_N:=\arg\min_{t=2,\ldots,T}\frac{1}{w(t)}\sum_{i=1}^N\sum_{s=1}^t(Y_{i,s}-\widebar{Y}_{i,t})^2,
\end{equation}
where $\{w(t)\}_{t=2}^T$ is a~sequence of weights specified in~\cite{PP2016}.

Since the panels are considered to be independent and the number of panels may be sufficiently large, one can estimate the correlation structure of the errors $[\eps_{1,1},\ldots,\eps_{1,T}]^{\top}$ empirically. We base the errors' estimates on \emph{residuals}
\begin{equation}\label{widehate}
\widehat{e}_{i,t}:=\left\{
\begin{array}{ll}
Y_{i,t}-\widebar{Y}_{i,\widehat{\tau}_N},& t\leq\widehat{\tau}_N,\\
Y_{i,t}-\widetilde{Y}_{i,\widehat{\tau}_N},& t>\widehat{\tau}_N.
\end{array}
\right.
\end{equation}

Then, the empirical version of the autocorrelation function is
\[
\widehat{\rho}_t:=\frac{1}{\widehat{\sigma}^2 NT}\sum_{i=1}^N\sum_{s=1}^{T-t}\widehat{e}_{i,s}\widehat{e}_{i,s+t}.
\]
Consequently, the kernel estimation of the cumulative autocorrelation function and shifted cumulative correlation function is adopted in lines with~\cite{Andrews1991}:
\begin{align*}
\widehat{r}(t)&=\sum_{|s|<t}(t-|s|)\kappa\left(\frac{s}{h}\right)\widehat{\rho}_s,\\
\widehat{R}(t,v)&=\sum_{s=1}^t\sum_{u=t+1}^v\kappa\left(\frac{u-s}{h}\right)\widehat{\rho}_{u-s},\quad t<v;
\end{align*}
where $h>0$ stands for the window size and $\kappa$ belongs to a~class of kernels
\begin{multline*}
\Big\{\kappa(\cdot):\,\mathbbm{R}\to[-1,1]\,\big|\,\kappa(0)=1,\,\kappa(x)=\kappa(-x),\,\forall x,\,\int_{-\infty}^{+\infty}\kappa^2(x)\ud x<\infty,\Big.\\
\Big.\kappa(\cdot)\,\mbox{ is continuos at $0$ and at all but a finite number of other points}\Big\}.
\end{multline*}

Since the variance parameter $\sigma$ simply cancels out from the limiting distribution of Theorem~\ref{underNull}, it neither has to be estimated nor known. Nevertheless, one can use $\widehat{\sigma}^2:=\frac{1}{NT}\sum_{i=1}^{N}\sum_{s=1}^{T}\widehat{e}_{i,s}^2$.

\section{Simulations}\label{sec:simul}
A~simulation experiment was performed to study the \emph{finite sample} properties of the test statistics for a~common change in panel means. In particular, the interest lies in the empirical \emph{sizes} of the proposed tests (i.e., based on $\mathcal{R}_{N}(T)$ and $\mathcal{S}_{N}(T)$) under the null hypothesis and in the empirical \emph{rejection} rate (power) under the alternative. Random samples of panel data ($5000$ each time) are generated from the panel change point model~\eqref{model}. The panel size is set to $T=10$ and $T=25$ in order to demonstrate the performance of the testing approaches in case of small and intermediate panel length. The number of panels considered is $N=50$ and $N=200$.

The correlation structure within each panel is modeled via random vectors generated from iid, AR(1), and GARCH(1,1) sequences. The considered AR(1) process has coefficient $\phi=0.3$. In case of GARCH(1,1) process, we use coefficients $\alpha_0=1$, $\alpha_1=0.1$, and $\beta_1=0.2$, which according to \citet[Example~1]{Lindner2009} gives a~strictly stationary process. In all three sequences, the innovations are obtained as iid random variables from a~standard normal $\mathsf{N}(0,1)$ or Student $t_5$ distribution. Simulation scenarios are produced as all possible combinations of the above mentioned settings.

When using the asymptotic distribution from Theorem~\ref{underNull}, the covariance matrix is estimated as proposed in Section~\ref{sec:covest} using the Parzen kernel
\[
\kappa_{P}(x)=\left\{\begin{array}{ll}
1-6x^2+6|x|^3, & 0\leq|x|\leq 1/2;\\
2(1-|x|)^3, & 1/2\leq|x|\leq 1;\\
0, & \mbox{otherwise}.
\end{array}
\right.
\]

Bartlett window as well

Several values of the smoothing window width $h$ are tried from the interval $[2,5]$ and all of them work fine providing comparable results. To simulate the asymptotic distribution of the test statistics, $2000$ multivariate random vectors are generated using the pre-estimated covariance matrix. To access the theoretical results under $H_0$ numerically, Table~\ref{tab:H0} provides the empirical specificity (one minus size) of the asymptotic tests based on $\mathcal{R}_{N}(T)$ and $\mathcal{S}_{N}(T)$, where the significance level is $\alpha=5\%$.
\begin{table}[!ht]
\caption{Empirical specificity ($1-$size) of the test under $H_0$ for test statistics \colorbox{sedaseda}{$\mathcal{R}_{N}(T)$} and \colorbox{seda}{$\mathcal{C}_{N}(T)$} using the asymptotic critical values, considering a~significance level of $5\%$, weight function $w(t)=t^2$, and smoothing window width $h=2$}
\label{tab:H0}
\begin{center}
\begin{tabular}{cccx{.85cm}x{.85cm}x{.85cm}x{.85cm}x{.85cm}x{.85cm}}
\toprule
$T$ & $N$ & innovations & \multicolumn{2}{c}{IID} & \multicolumn{2}{c}{AR(1)} & \multicolumn{2}{c}{GARCH(1,1)} \tabularnewline
\midrule
$10$ & $50$ & $\mathsf{N}(0,1)$ & \cellcolor[gray]{0.8}{$.948$} & \cellcolor[gray]{0.9}{$.934$} & \cellcolor[gray]{0.8}{$.933$} & \cellcolor[gray]{0.9}{$.823$} & \cellcolor[gray]{0.8}{$.946$} & \cellcolor[gray]{0.9}{$.935$} \tabularnewline
 & & $t_5$ & \cellcolor[gray]{0.8}{$.951$} & \cellcolor[gray]{0.9}{$.925$} & \cellcolor[gray]{0.8}{$.932$} & \cellcolor[gray]{0.9}{$.822$} & \cellcolor[gray]{0.8}{$.946$} & \cellcolor[gray]{0.9}{$.929$} \tabularnewline
 & $200$ & $\mathsf{N}(0,1)$ & \cellcolor[gray]{0.8}{$.950$} & \cellcolor[gray]{0.9}{$.933$} & \cellcolor[gray]{0.8}{$.939$} & \cellcolor[gray]{0.9}{$.825$} & \cellcolor[gray]{0.8}{$.950$} & \cellcolor[gray]{0.9}{$.938$} \tabularnewline
 & & $t_5$ & \cellcolor[gray]{0.8}{$.948$} & \cellcolor[gray]{0.9}{$.927$} & \cellcolor[gray]{0.8}{$.935$} & \cellcolor[gray]{0.9}{$.821$} & \cellcolor[gray]{0.8}{$.948$} & \cellcolor[gray]{0.9}{$.937$} \tabularnewline
\cmidrule(){1-9}
$25$ & $50$ & $\mathsf{N}(0,1)$ & \cellcolor[gray]{0.8}{$.946$} & \cellcolor[gray]{0.9}{$.940$} & \cellcolor[gray]{0.8}{$.932$} & \cellcolor[gray]{0.9}{$.780$} & \cellcolor[gray]{0.8}{$.947$} & \cellcolor[gray]{0.9}{$.945$} \tabularnewline
 & & $t_5$ & \cellcolor[gray]{0.8}{$.948$} & \cellcolor[gray]{0.9}{$.945$} & \cellcolor[gray]{0.8}{$.932$} & \cellcolor[gray]{0.9}{$.790$} & \cellcolor[gray]{0.8}{$.946$} & \cellcolor[gray]{0.9}{$.943$} \tabularnewline
 & $200$ & $\mathsf{N}(0,1)$ & \cellcolor[gray]{0.8}{$.949$} & \cellcolor[gray]{0.9}{$.939$} & \cellcolor[gray]{0.8}{$.930$} & \cellcolor[gray]{0.9}{$.801$} & \cellcolor[gray]{0.8}{$.951$} & \cellcolor[gray]{0.9}{$.939$} \tabularnewline
 & & $t_5$ & \cellcolor[gray]{0.8}{$.953$} & \cellcolor[gray]{0.9}{$.941$} & \cellcolor[gray]{0.8}{$.931$} & \cellcolor[gray]{0.9}{$.813$} & \cellcolor[gray]{0.8}{$.952$} & \cellcolor[gray]{0.9}{$.946$} \tabularnewline
\bottomrule
\end{tabular}
\end{center}
\end{table}

It may be seen that both approaches are close to the theoretical value of specificity $.95$. As expected, the best results are achieved in case of independence within the panel, because there is no information overlap between two consecutive observations. The precision of not rejecting the null is increasing as the number of panels is getting higher and the panel is getting longer as well.

The performance of both testing procedures under $H_1$ in terms of the empirical rejection rates is shown in Table~\ref{tab:H1}, where the change point is set to $\tau=\lfloor T/2 \rfloor$ and the change sizes $\delta_i$ are independently uniform on $[1,3]$ in $33\%$, $66\%$ or in all panels.
\begin{table}[!ht]
\caption{Empirical sensitivity (power) of the test under $H_1$ for test statistics \colorbox{sedaseda}{$\mathcal{R}_{N}(T)$} and \colorbox{seda}{$\mathcal{C}_{N}(T)$} using the asymptotic critical values, considering a~significance level of $5\%$, weight function $w(t)=t^2$, and smoothing window width $h=2$}
\label{tab:H1}
\begin{center}
\begin{tabular}{ccccx{.95cm}x{.95cm}x{.95cm}x{.95cm}x{.95cm}x{.95cm}}
\toprule
$H_1$ & $T$ & $N$ & innovations & \multicolumn{2}{c}{IID} & \multicolumn{2}{c}{AR(1)} & \multicolumn{2}{c}{GARCH(1,1)} \tabularnewline
\midrule
$33\%$ & $10$ & $50$ & $\mathsf{N}(0,1)$ & \cellcolor[gray]{0.8}{$.235$} & \cellcolor[gray]{0.9}{$1.000$} & \cellcolor[gray]{0.8}{$.256$} & \cellcolor[gray]{0.9}{$.999$} & \cellcolor[gray]{0.8}{$.193$} & \cellcolor[gray]{0.9}{$1.000$} \tabularnewline
& & & $t_5$ & \cellcolor[gray]{0.8}{$.174$} & \cellcolor[gray]{0.9}{$.999$} & \cellcolor[gray]{0.8}{$.202$} & \cellcolor[gray]{0.9}{$.996$} & \cellcolor[gray]{0.8}{$.201$} & \cellcolor[gray]{0.9}{$.999$} \tabularnewline
& & $200$ & $\mathsf{N}(0,1)$ & \cellcolor[gray]{0.8}{$.453$} & \cellcolor[gray]{0.9}{$1.000$} & \cellcolor[gray]{0.8}{$.486$} & \cellcolor[gray]{0.9}{$1.000$} & \cellcolor[gray]{0.8}{$.387$} & \cellcolor[gray]{0.9}{$1.000$} \tabularnewline
& & & $t_5$ & \cellcolor[gray]{0.8}{$.360$} & \cellcolor[gray]{0.9}{$1.000$} & \cellcolor[gray]{0.8}{$.393$} & \cellcolor[gray]{0.9}{$1.000$} & \cellcolor[gray]{0.8}{$.389$} & \cellcolor[gray]{0.9}{$1.000$} \tabularnewline
\cmidrule(l){2-10}
& $25$ & $50$ & $\mathsf{N}(0,1)$ & \cellcolor[gray]{0.8}{$.376$} & \cellcolor[gray]{0.9}{$1.000$} & \cellcolor[gray]{0.8}{$.394$} & \cellcolor[gray]{0.9}{$.992$} & \cellcolor[gray]{0.8}{$.312$} & \cellcolor[gray]{0.9}{$1.000$} \tabularnewline
& & & $t_5$ & \cellcolor[gray]{0.8}{$.294$} & \cellcolor[gray]{0.9}{$1.000$} & \cellcolor[gray]{0.8}{$.301$} & \cellcolor[gray]{0.9}{$.993$} & \cellcolor[gray]{0.8}{$.312$} & \cellcolor[gray]{0.9}{$1.000$} \tabularnewline
& & $200$ & $\mathsf{N}(0,1)$ & \cellcolor[gray]{0.8}{$.685$} & \cellcolor[gray]{0.9}{$1.000$} & \cellcolor[gray]{0.8}{$.699$} & \cellcolor[gray]{0.9}{$.995$} & \cellcolor[gray]{0.8}{$.584$} & \cellcolor[gray]{0.9}{$1.000$} \tabularnewline
& & & $t_5$ & \cellcolor[gray]{0.8}{$.561$} & \cellcolor[gray]{0.9}{$1.000$} & \cellcolor[gray]{0.8}{$.565$} & \cellcolor[gray]{0.9}{$1.000$} & \cellcolor[gray]{0.8}{$.590$} & \cellcolor[gray]{0.9}{$1.000$} \tabularnewline
\cmidrule(){1-10}

$66\%$ & $10$ & $50$ & $\mathsf{N}(0,1)$ & \cellcolor[gray]{0.8}{$.450$} & \cellcolor[gray]{0.9}{$1.000$} & \cellcolor[gray]{0.8}{$.491$} & \cellcolor[gray]{0.9}{$1.000$} & \cellcolor[gray]{0.8}{$.386$} & \cellcolor[gray]{0.9}{$1.000$} \tabularnewline
& & & $t_5$ & \cellcolor[gray]{0.8}{$.360$} & \cellcolor[gray]{0.9}{$1.000$} & \cellcolor[gray]{0.8}{$.377$} & \cellcolor[gray]{0.9}{$1.000$} & \cellcolor[gray]{0.8}{$.390$} & \cellcolor[gray]{0.9}{$1.000$} \tabularnewline
& & $200$ & $\mathsf{N}(0,1)$ & \cellcolor[gray]{0.8}{$.774$} & \cellcolor[gray]{0.9}{$1.000$} & \cellcolor[gray]{0.8}{$.807$} & \cellcolor[gray]{0.9}{$1.000$} & \cellcolor[gray]{0.8}{$.677$} & \cellcolor[gray]{0.9}{$1.000$} \tabularnewline
& & & $t_5$ & \cellcolor[gray]{0.8}{$.642$} & \cellcolor[gray]{0.9}{$1.000$} & \cellcolor[gray]{0.8}{$.692$} & \cellcolor[gray]{0.9}{$1.000$} & \cellcolor[gray]{0.8}{$.688$} & \cellcolor[gray]{0.9}{$1.000$} \tabularnewline
\cmidrule(l){2-10}
& $25$ & $50$ & $\mathsf{N}(0,1)$ & \cellcolor[gray]{0.8}{$.688$} & \cellcolor[gray]{0.9}{$1.000$} & \cellcolor[gray]{0.8}{$.694$} & \cellcolor[gray]{0.9}{$1.000$} & \cellcolor[gray]{0.8}{$.581$} & \cellcolor[gray]{0.9}{$1.000$} \tabularnewline
& & & $t_5$ & \cellcolor[gray]{0.8}{$.558$} & \cellcolor[gray]{0.9}{$1.000$} & \cellcolor[gray]{0.8}{$.570$} & \cellcolor[gray]{0.9}{$1.000$} & \cellcolor[gray]{0.8}{$.594$} & \cellcolor[gray]{0.9}{$1.000$} \tabularnewline
& & $200$ & $\mathsf{N}(0,1)$ & \cellcolor[gray]{0.8}{$.951$} & \cellcolor[gray]{0.9}{$1.000$} & \cellcolor[gray]{0.8}{$.959$} & \cellcolor[gray]{0.9}{$1.000$} & \cellcolor[gray]{0.8}{$.905$} & \cellcolor[gray]{0.9}{$1.000$} \tabularnewline
& & & $t_5$ & \cellcolor[gray]{0.8}{$.874$} & \cellcolor[gray]{0.9}{$1.000$} & \cellcolor[gray]{0.8}{$.888$} & \cellcolor[gray]{0.9}{$1.000$} & \cellcolor[gray]{0.8}{$.906$} & \cellcolor[gray]{0.9}{$1.000$} \tabularnewline
\cmidrule(){1-10}

$100\%$ & $10$ & $50$ & $\mathsf{N}(0,1)$ & \cellcolor[gray]{0.8}{$.641$} & \cellcolor[gray]{0.9}{$1.000$} & \cellcolor[gray]{0.8}{$.667$} & \cellcolor[gray]{0.9}{$1.000$} & \cellcolor[gray]{0.8}{$.563$} & \cellcolor[gray]{0.9}{$1.000$} \tabularnewline
& & & $t_5$ & \cellcolor[gray]{0.8}{$.519$} & \cellcolor[gray]{0.9}{$1.000$} & \cellcolor[gray]{0.8}{$.547$} & \cellcolor[gray]{0.9}{$1.000$} & \cellcolor[gray]{0.8}{$.546$} & \cellcolor[gray]{0.9}{$1.000$} \tabularnewline
& & $200$ & $\mathsf{N}(0,1)$ & \cellcolor[gray]{0.8}{$.928$} & \cellcolor[gray]{0.9}{$1.000$} & \cellcolor[gray]{0.8}{$.945$} & \cellcolor[gray]{0.9}{$1.000$} & \cellcolor[gray]{0.8}{$.868$} & \cellcolor[gray]{0.9}{$1.000$} \tabularnewline
& & & $t_5$ & \cellcolor[gray]{0.8}{$.844$} & \cellcolor[gray]{0.9}{$1.000$} & \cellcolor[gray]{0.8}{$.869$} & \cellcolor[gray]{0.9}{$1.000$} & \cellcolor[gray]{0.8}{$.872$} & \cellcolor[gray]{0.9}{$1.000$} \tabularnewline
\cmidrule(l){2-10}
& $25$ & $50$ & $\mathsf{N}(0,1)$ & \cellcolor[gray]{0.8}{$.873$} & \cellcolor[gray]{0.9}{$1.000$} & \cellcolor[gray]{0.8}{$.884$} & \cellcolor[gray]{0.9}{$1.000$} & \cellcolor[gray]{0.8}{$.792$} & \cellcolor[gray]{0.9}{$1.000$} \tabularnewline
& & & $t_5$ & \cellcolor[gray]{0.8}{$.760$} & \cellcolor[gray]{0.9}{$1.000$} & \cellcolor[gray]{0.8}{$.771$} & \cellcolor[gray]{0.9}{$1.000$} & \cellcolor[gray]{0.8}{$.789$} & \cellcolor[gray]{0.9}{$1.000$} \tabularnewline
& & $200$ & $\mathsf{N}(0,1)$ & \cellcolor[gray]{0.8}{$.997$} & \cellcolor[gray]{0.9}{$1.000$} & \cellcolor[gray]{0.8}{$.997$} & \cellcolor[gray]{0.9}{$1.000$} & \cellcolor[gray]{0.8}{$.985$} & \cellcolor[gray]{0.9}{$1.000$} \tabularnewline
& & & $t_5$ & \cellcolor[gray]{0.8}{$.977$} & \cellcolor[gray]{0.9}{$1.000$} & \cellcolor[gray]{0.8}{$.982$} & \cellcolor[gray]{0.9}{$1.000$} & \cellcolor[gray]{0.8}{$.986$} & \cellcolor[gray]{0.9}{$1.000$} \tabularnewline
\bottomrule
\end{tabular}
\end{center}
\end{table}

One can conclude that the power of both tests increases as the panel size and the number of panels increase, which is straightforward and expected. Moreover, higher power is obtained when a~larger portion of panels is subject to have a~change in mean. The test power drops when switching from independent observations within the panel to dependent ones. Innovations with heavier tails (i.e., $t_5$) yield smaller power than innovations with lighter tails. Generally, the newly defined test statistic $\mathcal{S}_{N}(T)$ \emph{outperforms} $\mathcal{R}_{N}(T)$ in all scenarios with respect to the power.


Finally, an~early change point is discussed very briefly. We stay with standard normal innovations, iid observations within the panel, the size of changes $\delta_i$ being independently uniform on $[1,3]$ in all panels, and the change point is $\tau=3$ in case of $T=10$ and $\tau=5$ for $T=25$. The empirical sensitivities of both tests for small values of $\tau$ are shown in Table~\ref{tab:tau}.
\begin{table}[!ht]
\caption{Empirical sensitivity of the test for small values of $\tau$ under $H_1$ for test statistics \colorbox{sedaseda}{$\mathcal{R}_{N}(T)$} and \colorbox{seda}{$\mathcal{C}_{N}(T)$} using the asymptotic critical values, considering a~significance level of $5\%$, weight function $w(t)=t^2$, and smoothing window width $h=2$}
\label{tab:tau}
\begin{center}
\begin{tabular}{x{.55cm}x{.55cm}x{.65cm}x{.95cm}x{.95cm}x{.55cm}x{.55cm}x{.65cm}x{.95cm}x{.95cm}}
\toprule
$T$ & $\tau$ & $N$ & \multicolumn{2}{c}{$H_1$, iid, $\mathsf{N}(0,1)$} & $T$ & $\tau$ & $N$ & \multicolumn{2}{c}{$H_1$, iid, $\mathsf{N}(0,1)$} \tabularnewline
\midrule
$10$ & $3$ & $50$  & \cellcolor[gray]{0.8}{$.551$} & \cellcolor[gray]{0.9}{$1.000$} & $25$ & $5$ & $50$ & \cellcolor[gray]{0.8}{$.629$} & \cellcolor[gray]{0.9}{$1.000$} \tabularnewline
 & & $200$ & \cellcolor[gray]{0.8}{$.867$} & \cellcolor[gray]{0.9}{$1.000$} & & & $200$ & \cellcolor[gray]{0.8}{$.927$} & \cellcolor[gray]{0.9}{$1.000$} \tabularnewline
\bottomrule
\end{tabular}
\end{center}
\end{table}

When the change point is not in the middle of the panel, the power of the test generally falls down. The source of such decrease is that the left or right part of the panel possesses less observations with constant mean, which leads to a~decrease of precision in the correlation estimation. Nevertheless, $\mathcal{S}_{N}(T)$ again outperforms $\mathcal{R}_{N}(T)$ even for early or late change points (the late change points are not numerically demonstrated here).

\section{Conclusions}\label{sec:concl}
In this paper, we consider the change point problem in panel data with fixed panel size. Occurrence of common breaks in panel means is tested. We introduce a~ratio type test statistic and derive its asymptotic properties. Under the null hypothesis of no change, the test statistic weakly converges to a~functional of the multivariate normal random vector with zero mean and covariance structure depending on the intra-panel covariances. As shown in the paper, these covariances can be estimated and, consequently, used for testing whether a~change in means occurred or not. This is indeed feasible, because the test statistic under the alternative converges to infinity in probability. Furthermore, the whole stochastic theory behind requires relatively simple assumptions, which are not too restrictive.

A~simulation study illustrates that even for small panel size, both investigated approaches---the newly derived one based on $\mathcal{S}_{N}(T)$ and the older one proposed in~\cite{PP2015}---work fine. One may judge that both methods keep the significance level under the null, while various simulation scenarios are considered. Besides that, the power of the test is higher in case of $\mathcal{S}_{N}(T)$ compared to $\mathcal{R}_{N}(T)$. Finally, the proposed method is applied to insurance data, for which the change point analysis in panel data provides an~appealing approach.

\subsection{Discussion}\label{subsec:disc}
Our setup can be modified by considering large panel size, i.e., $T\to\infty$. Consequently, the whole theory leads to convergences to functionals of Gaussian processes with a~covariance structure derived in a~very similar fashion as for fixed $T$. However, our motivation is to develop techniques for fixed and relatively small panel size.

Dependent panels may be taken into account and the presented work might be generalized for some kind of asymptotic independence over the panels or prescribed dependence among the panels. Nevertheless, our incentive is determined by a~problem from non-life insurance, where the association of insurance companies consists of a~relatively high number of insurance companies. Thus, the portfolio of yearly claims is so diversified, that the panels corresponding to insurance companies' yearly claims may be viewed as independent and neither natural ordering nor clustering has to be assumed.

\appendix
\section{Proofs}
\begin{proof}[Proof of Theorem~\ref{underNull}]
Let us define
\[
U_N(t):=\frac{1}{\sigma\sqrt{N}}\sum_{i=1}^N\sum_{s=1}^t(Y_{i,s}-\mu_i).
\]
Using the multivariate Lindeberg-L\'{e}vy CLT for a~sequence of $T$-dimensional iid random vectors $\{[\sum_{s=1}^1\eps_{i,s},\ldots,\sum_{s=1}^T\eps_{i,s}]^{\top}\}_{i\in\mathbbm{N}}$, we have under $H_0$
\[
[U_N(1),\ldots,U_N(T)]^{\top}\xrightarrow[N\to\infty]{\dist}[X_1,\ldots,X_T]^{\top},
\]
since $\Var[\sum_{s=1}^1\eps_{1,s},\ldots,\sum_{s=1}^T\eps_{1,s}]^{\top}=\bLambda$. Indeed, the $t$-th diagonal element of the covariance matrix $\bLambda$ is $\Var\sum_{s=1}^t\eps_{1,s}=r(t)$ and the upper off-diagonal element on position $(t,v)$ is
\begin{align*}
\Cov\left(\sum_{s=1}^t\eps_{1,s},\sum_{u=1}^v\eps_{1,u}\right)&=\Var\sum_{s=1}^t\eps_{1,s}+\Cov\left(\sum_{s=1}^t\eps_{1,s},\sum_{u=t+1}^v\eps_{1,u}\right)\\
&=r(t)+R(t,v),\quad t<v.
\end{align*}
Moreover, let us define the reverse analogue to $U_N(t)$, i.e.,
\[
V_N(t):=\frac{1}{\sigma\sqrt{N}}\sum_{i=1}^N\sum_{s=t+1}^T(Y_{i,s}-\mu_i)=U_N(T)-U_N(t).
\]
Hence,
\begin{align*}
U_{N}(s)-\frac{s}{t}U_{N}(t)&=\frac{1}{\sigma\sqrt{N}}\sum_{i=1}^N\left\{\sum_{r=1}^s\left[\left(Y_{i,r}-\mu_i\right)-\frac{1}{t}\sum_{v=1}^t \left(Y_{i,v}-\mu_i\right)\right]\right\}\\
&=\frac{1}{\sigma\sqrt{N}}\sum_{i=1}^N\sum_{r=1}^s\left(Y_{i,r}-\widebar{Y}_{i,t}\right)
\end{align*}
and, consequently,
\begin{align*}
V_{N}(s)-\frac{T-s}{T-t}V_{N}(t)&=\frac{1}{\sigma\sqrt{N}}\sum_{i=1}^N\left\{\sum_{r=s+1}^T\left[\left(Y_{i,r}-\mu_i\right)-\frac{1}{T-t}\sum_{v=t+1}^T \left(Y_{i,v}-\mu_i\right)\right]\right\}\\
&=\frac{1}{\sigma\sqrt{N}}\sum_{i=1}^N\sum_{r=s+1}^T\left(Y_{i,r}-\widetilde{Y}_{i,t}\right).
\end{align*}
Using the Cram\'{e}r-Wold device, we end up with
\[
\max_{t=1,\ldots,T-1}\left|U_{N}(t)-\frac{t}{T}U_{N}(T)\right|\xrightarrow[N\to\infty]{\dist}\max_{t=1,\ldots,T-1}\left|X_t-\frac{t}{T}X_T\right|
\]
and
\begin{multline*}
\max_{t=2,\ldots,T-2}\frac{\max_{s=1,\ldots,t}\left|U_{N}(s)-\frac{s}{t}U_{N}(t)\right|}{\max_{s=t,\ldots,T-1}\left|V_{N}(s)-\frac{T-s}{T-t}V_{N}(t)\right|}\\
\xrightarrow[N\to\infty]{\dist}\max_{t=2,\ldots,T-2}\frac{\max_{s=1,\ldots,t}\left|X_s-\frac{s}{t}X_t\right|}{\max_{s=t,\ldots,T-1}\left|(X_T-X_s)-\frac{T-s}{T-t}(X_T-X_t)\right|}.
\end{multline*}
\end{proof}

\begin{proof}[Proof of Theorem~\ref{underAlternative}]
Considering $\mathcal{C}_N(T)$, we have under alternative $H_1$ that
\begin{align*}
&\frac{1}{\sigma\sqrt{N}}\left|\sum_{i=1}^N\sum_{s=1}^{\tau}\left(Y_{i,s}-\widebar{Y}_{i,T}\right)\right|=\frac{1}{\sigma\sqrt{N}}\left|\sum_{i=1}^N\sum_{s=1}^{\tau}\left(\mu_i+\sigma\eps_{i,s}-\frac{1}{T}\sum_{v=1}^{T}(\mu_i+\sigma\eps_{i,v})-\frac{1}{T}\delta_i\right)\right|\\
&=\frac{1}{\sqrt{N}}\left|\sum_{i=1}^N\sum_{s=1}^{\tau}\left(\eps_{i,s}-\widebar{\eps}_{i,T}\right)-\frac{\tau}{\sigma T}\sum_{i=1}^N\delta_i\right|=\mathcal{O}_{\prob}(1)+\frac{\tau}{\sigma T\sqrt{N}}\left|\sum_{i=1}^N\delta_i\right|\xrightarrow[]{\prob}\infty,\quad N\to\infty,
\end{align*}
where $\widebar{\eps}_{i,T}=\frac{1}{\tau}\sum_{v=1}^{T}\eps_{i,v}$.

In case of $\mathcal{R}_N(T)$, let $t=\tau+1$. Then under alternative $H_1$, it holds that
\begin{align*}
&\frac{1}{\sigma\sqrt{N}}\max_{s=1,\ldots,\tau+1}\left|\sum_{i=1}^N\sum_{r=1}^s\left(Y_{i,r}-\widebar{Y}_{i,\tau+1}\right)\right|\geq\frac{1}{\sigma\sqrt{N}}\left|\sum_{i=1}^N\sum_{r=1}^{\tau}\left(Y_{i,r}-\widebar{Y}_{i,\tau+1}\right)\right|\\
&=\frac{1}{\sigma\sqrt{N}}\left|\sum_{i=1}^N\sum_{r=1}^{\tau}\left(\mu_i+\sigma\eps_{i,r}-\frac{1}{\tau+1}\sum_{v=1}^{\tau+1}(\mu_i+\sigma\eps_{i,v})-\frac{1}{\tau+1}\delta_i\right)\right|\\
&=\frac{1}{\sqrt{N}}\left|\sum_{i=1}^N\sum_{r=1}^{\tau}\left(\eps_{i,r}-\widebar{\eps}_{i,\tau+1}\right)-\frac{\tau}{\sigma(\tau+1)}\sum_{i=1}^N\delta_i\right|=\mathcal{O}_{\prob}(1)+\frac{\tau}{\sigma(\tau+1)\sqrt{N}}\left|\sum_{i=1}^N\delta_i\right|\xrightarrow[]{\prob}\infty,\quad N\to\infty,
\end{align*}
where $\widebar{\eps}_{i,\tau+1}=\frac{1}{\tau+1}\sum_{v=1}^{\tau+1}\eps_{i,v}$.

Since there is no change after $\tau+1$ and $\tau\leq T-3$, then by Theorem~\ref{underNull} we obtain
\[
\frac{1}{\sigma\sqrt{N}}\max_{s=\tau+1,\ldots,T-1}\left|\sum_{i=1}^N\sum_{r=s+1}^T\left(Y_{i,r}-\widetilde{Y}_{i,\tau+1}\right)\right|\xrightarrow[N\to\infty]{\dist}\max_{s=\tau+1,\ldots,T-1}\left|Z_s-\frac{T-s}{T-\tau}Z_{\tau+1}\right|.
\]
\end{proof}

\section*{Acknowledgment}
Institutional support to Barbora Pe\v{s}tov\'{a} was provided by RVO:67985807. The research of Michal Pe\v{s}ta was supported by the Czech Science Foundation project ``DYME -- Dynamic Models in Economics'' No.~P402/12/G097.


% ...

\end{document}